\documentclass[review,12pt]{elsarticle}
\usepackage{microtype}

\usepackage[boxruled]{algorithm2e}
\usepackage{algorithmic}
\usepackage{graphicx}
\usepackage{graphics}
\usepackage{epsfig}
\usepackage{epstopdf}
\usepackage{latexsym}

\usepackage{amssymb}
\usepackage{amsthm}

\newtheorem{theorem}{Theorem}[section]
\newtheorem{definition}[theorem]{Definition}
\newtheorem{lemma}[theorem]{Lemma}
\newtheorem{corollary}[theorem]{Corollary}

\journal{Journal of Computer and System Sciences}

\begin{document}

\begin{frontmatter}




\title{On the Connectivity Preserving Minimum  Cut Problem}


\author[char]{Qi Duan}
\ead{qduan@uncc.edu}

\cortext[char]{Corresponding author, phone (704)524-9768}

\author[buf]{Jinhui Xu}
\ead{jinhui@buffalo.edu}

\address[char]{Department of Software and Information Systems\\
University of North Carolina at Charlotte \\
Charlotte, NC 28223, USA}

\address[buf]{Department of Computer Science and Engineering \\
  State University of New York at Buffalo\\
  Buffalo, NY 14260, USA}

\begin{abstract}
In this paper, we study a generalization of the classical minimum cut problem, called {\em Connectivity Preserving Minimum Cut (CPMC)} problem,
which seeks a minimum cut to separate a pair (or pairs) of source and destination nodes and meanwhile ensure the connectivity between the source and its partner node(s). The CPMC problem is a rather powerful formulation for a set of problems and finds applications in many other areas, such as network security, image processing, data mining, pattern recognition, and machine learning. For this important problem, we consider two variants, connectivity preserving minimum node cut (CPMNC) and connectivity preserving minimum edge cut (CPMEC). For CPMNC, we show that it cannot be approximated within $\alpha logn$ for some constant $\alpha$ unless $P$=$NP$, and cannot be  approximated within any $poly(logn)$ unless $NP$ has quasi-polynomial time algorithms. The hardness results hold even for graphs with unit weight and bipartite graphs.  Particularly,
we show that  polynomial time solutions exist for CPMEC in planar graphs and for CPMNC in some special planar graphs.
The hardness of CPMEC in general graphs remains open, but the polynomial time
 algorithm in  planar graphs still has important
practical applications.
\end{abstract}

\begin{keyword}

Minimum Cut; Inapproximability; Connectivity Preserving
\end{keyword}

\end{frontmatter}


\section{Introduction} 

Minimum cut is one of the most fundamental problems in computer science and has numerous applications in many different areas~\cite{Papadimitriou93,Vazirani04,PS98,Lawler01}. 
In this paper, we consider a new generalization of the minimum cut problem, called {\em connectivity preserving minimum cut (CPMC)} problem arising in several areas. In this problem, we are given a connected graph $G=(V,E)$ with positive node (or edge) weights, a source node $s_1$ and its partner node $s_2$, and a destination node $t$. The objective is to compute a cut with minimum weight to disconnect the source $s_1$ and destination $t$, and meanwhile preserve the connectivity of $s_1$ and its partner node $s_2$ (i.e., $s_{1}$ and $s_{2}$ are connected after the cut). The weights can be associated with either the nodes (i.e., vertices) or the edges, and accordingly the cut can be either a set of nodes, called a connectivity preserving node cut, or a set of edges, called a connectivity preserving edge cut.
Corresponding to the two types of cuts, the CPMC problem has two variants, {\em connectivity preserving minimum node cut (CPMNC)} and {\em connectivity preserving minimum edge cut (CPMEC)}.  

The CPMC problem has both theoretical and practical importance. Theoretically, it is closely related to three fundamental problems, minimum cut, set cover, and shortest path.
Practically, the CPMC problem finds applications in many different areas. In 
network security, for example,  CPMC can be used to identify potential nodes for attacking. In such applications, an attacker (or police) may want to intercept all communication (or traffic) between a source node $s_1$ and a destination node $t$. It is possible that some nodes  with (direct) connection to the destination might already have been compromised. To maximally utilize such nodes,  the attacker only needs to compromise another set of nodes with minimum cost so that all traffic between the source and destination nodes passes one of the compromised nodes. 
To solve this problem, one can formulate it as a CPMC problem in which the compromised nodes are treated as partners of the source after removing their connections to the destination. 
In applications related to network reliability and emergency recovery, a node in a network might be contaminated, and has to be separated from some critical nodes. Meanwhile,
traffic flows among the critical nodes have to be maintained with a minimum cost. To solve such a problem, one can treat the critical nodes as the source and partner nodes and the contaminated node as the destination node, and formulates it as a CPMC problem. In data mining, machine learning, and image segmentation, CPMC can be used to model clustering or segmentation problems with additional constrains for clustering or segmenting certain objects together.

The CPMC problem can be generalized in several ways. For example, we may have multiple pairs of source and destination nodes, and each source node may have multiple partner nodes. The simplest version is the 3-node case in which only one source node, one destination node, and one partner node exist. 
Note that the 3-node case  is much different from the minimum 3-terminal cut problem~\cite{Dahlhaus94} in which all three nodes are required to be separated, whereas in the 3-node case two nodes are required to be connected. 
In this paper, we will  mainly focus on the 3-node case.


The CPMC problem is in general quite challenging, even for the 3-node case. 
One of the main reasons is that the connectivity preserving requirement and the minimum cut requirement seem to be contradicting to each other.  As  it will be shown later,  the hardness of the CPMC problem  increases dramatically with the added connectivity requirement. This phenomenon (i.e., increased hardness with the additional connectivity constraint) is consistent with the observations by Yannakakis \cite{Yannakakis79} in several other graph related optimization problems.
%

The CPMC problem is a new and interesting problem. To the best of our knowledge, it has not been studied previously. Related problems include the non-separating cycle 
 and optimal cycle problems in certain surfaces~\cite{EN11,VL02}. Since there is no restriction on the source and its partner nodes, CPMC seems to be  more 
general and fundamental.
   
In this paper, we mainly consider CPMNC, CPMEC, and  CPMC in planar graphs. For the CPMNC problem, we show that the problem is extremely hard to solve and to approximate, even for some very special cases.
Particularly, we show that it cannot be approximated within a factor of  $\alpha logn$ for some small constant $\alpha$ unless P=NP. We also use Feige and Lovasz's  two-prover one round interactive  proof protocol~\cite{FL92} to show that the CPMNC problem cannot be approximated within any  $poly(logn)$ factor unless $NP\subset DTIME(n^{poly(logn)})$. The hardness results hold even for unit-weighted graphs and bipartite graphs.

For planar graphs, we show that the CPMNC problem can be solved in polynomial time if $s_{1}$ and $s_{2}$  are on the same face.
For the CPMEC problem, we present a polynomial time solution for general planar graphs, which can be used for CPMC applications in image processing and machine learning. 
We also reveal a close relation between a Location Constrained Shortest Path (LCSP) problem and the CPMEC problem in special planar graphs in which  $s_1$ and $t$ are in the same face, and give polynomial time solutions to both problems. 




 
 
 


\vspace{-0.15in}
\section{Connectivity Preserving Node Cut Problem}
\label{sec-nc}
\vspace{-0.1in}


First we note that the CPMNC problem is an NP optimization problem. To determine whether a valid cut exists, one just needs to check if $t$ is connected to any bridge node between $s_1$ and $s_2$; if so, then no valid cut exists. Clearly, this can be done in polynomial time. Thus, we assume thereafter that a cut always exists.

We first define the decision version of the CPMNC problem.

\begin{definition} [Decision Problem of CPMNC]
Given an undirected graph $G=(V,E)$ with each node $v_i \in V$ associated with a positive integer weight $c_i$, a source node $s_1$, a partner node $s_2$, a destination node $t$, and an integer $b>0$, determine whether there exists a subset of nodes in $V$ with total weight less than or equal to $b$ such that the removal of this subset disconnects $t$ from $s_1$ but preserves the connectivity between $s_1$ and $s_2$.
\end{definition}

The decision version of the CPMEC problem can be defined similarly.  

\begin{theorem}
\label{the-nc1}
The CPMNC problem is NP-complete
 and cannot be approximated within 
$\alpha_1logn$ for some constant $\alpha_1$ 
unless $P=NP$, where $n=|V|$.  
\end{theorem} 

\begin{proof}
To prove the theorem, we reduce the set cover problem
to this problem. In the set cover  problem, we have a ground set $\mathcal{T}=\{e_1, e_2, \ldots, e_{n_1} \}$ of $n_1$ elements, and a set  $\mathcal{S}= \{S_1, S_2, \ldots, S_k \}$ of $k$ subsets of $\mathcal{T}$ with each $S_i \in \mathcal{S}$ associated with a weight $w_i$. The objective is to select a set $\mathcal{O}$ of subsets in $\mathcal{S}$ so that the union of all subsets in $\mathcal{O}$ contains every element in $\mathcal{T}$ and the total weight of subsets in $\mathcal{O}$ is minimized. 

Given an instance $I$ of the set cover problem with $n_1$ elements and $k$ sets, we construct a new graph. The new graph has an element gadget for every element, 
 and every element gadget contains $k_1+2$
nodes, where $k_1$ is the number
of sets that contains this element.
 In every gadget, there are two end
points, and $k_1$ internal
nodes  are connected to the two end nodes
in parallel. Every internal  nodes
of a gadget corresponds to a set that
contains this element.
 All such $n_1$
gadgets are connected sequentially
through their end points,
with $s_1$ and $s_2$ at the two ends
of the whole construction.     
All nodes correspond to the same
set are connected to a new node which we
call set node, and all set nodes are connected
to $t$.
Figure ~\ref{fig:set}
is the graph constructed for set cover
instance with three elements $x_1$, $x_2$, and
$x_3$, three sets $A_1=\{ x_1,x_3\}$,  $A_2=\{ x_2,x_3\}$,
and   $A_3=\{ x_1,x_2\}$. 

 \begin{figure}
\centering
\includegraphics[height=3in]{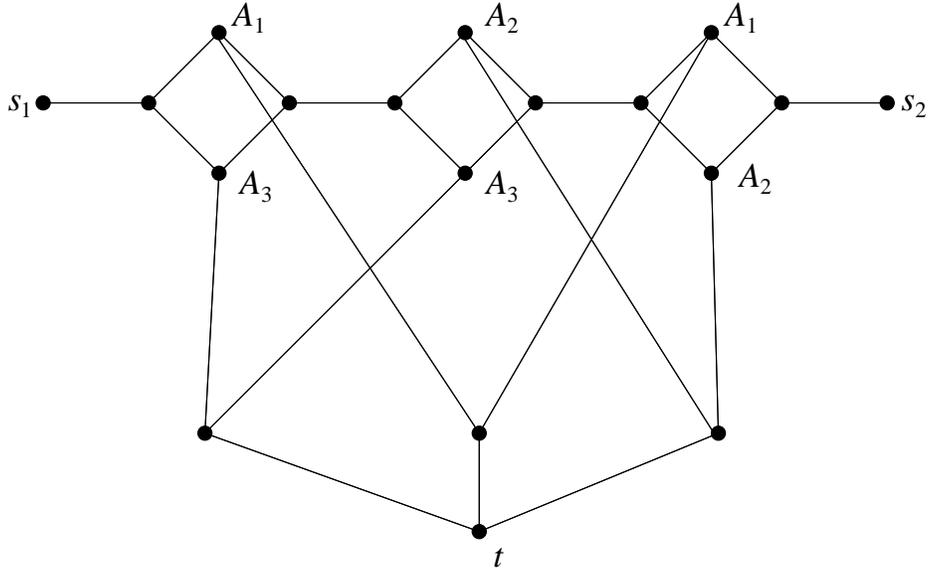}
\caption{ An example illustrating Theorem \ref{the-nc1}.}
\label{fig:set}
\end{figure}

\begin{figure}
\centering
\includegraphics[height=3in]{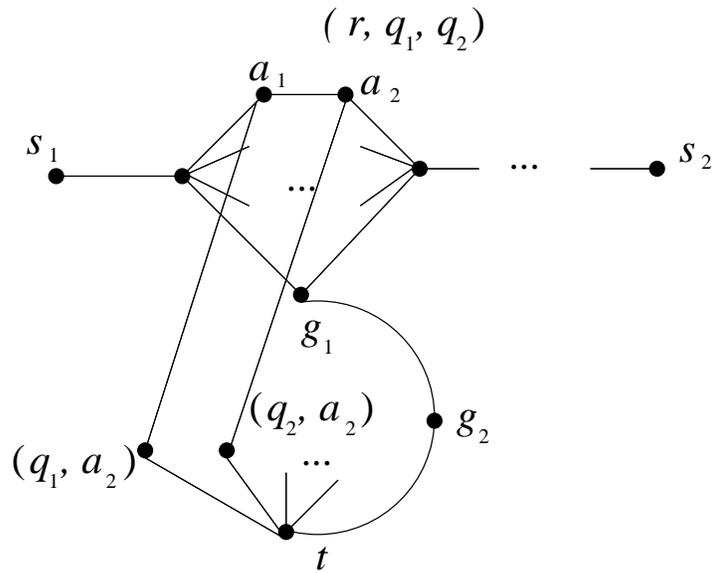}
\caption{ An example illustrating Theorem \ref{the-nc2}.}
\label{fig:mip}
\end{figure}

%


 Every set node is assigned a weight
$w_in_1k$, where $w_i$ is the weight of the
corresponding set in the original set cover instance.     
All other nodes are assigned 
weight $1$. We  let $b=n_1kD_1 + n_1k-1$,
where $D_1$ is the upper bound of weight
in the set cover instance. Note that one cannot
put all nodes into the cut  in an element
gadget, otherwise $s_1$ and $s_2$
will be separated. Now we can see that if the set
cover instance has a cover with weight
no more $D_1$, then we can choose the
following cut:  The cut contains those
set nodes contained in the  cover and 
all the gadget nodes which are not
in the set cover. The cut has a weight
 $n_1kD_1 + g_1$,  where $g_1<n_1k$. Similarly
if we can find a cut with weight no more than
  $n_1kD_1 + n_1k-1$, then we can find a corresponding
set cover with weight no more than $D_1$.
Furthermore, since set cover cannot
be approximated within $\alpha logn$ 
for some constant $\alpha$ unless NP=P~\cite{RS97,Feige98},
we can see that the connectivity preserving minimum cut problem
cannot be approximated within  $\alpha_1logn$
for some constant $\alpha_1$ unless NP=P.
Suppose the optimal solution of the set
cover instance is $D$, then the optimal solution
of the constructed graph has a minimum cut
with weight $n_1kD+g_2$, where $0<g_2<n_1k$.
If we can find a cut
in which the total weight (in the set cover instance) of  all set nodes
is $D_1$, then the cut has a weight  $n_1kD_1 + g_1$,
  where $0<g_1<n_1k$.
Assume   
$ \frac{n_1kD_1+g_2}{n_1kD + g_1} <\alpha_1 log(n_1k), $
for some $\alpha_1$,
 then we have
 $ \frac{D_1}{D} <\frac{n_1kD_1+g_2}{n_1kD + g_1} 
+ o(1) < \alpha_1log(n_1k).$

For the set cover problem with $n_1$ elements
and $k=poly(n_1)$ sets, it cannot be approximated
 within $\alpha logn_1$ unless 
NP=P~\cite{RS97,Feige98}. Since $k$ is bounded by some polynomial
in $n_1$,
 we can see  
   $ \frac{D_1}{D}  < \alpha_1log(n_1k)
  \leq \alpha_1\alpha_2logn_1 ,$
  where $\alpha_2$ is another constant. If we choose
$\alpha_1 \leq  \alpha/\alpha_2$, then 
    $ \frac{D_1}{D} \leq \alpha logn. $
Now we have a contradiction, which means that the problem
cannot be approximated within 
$\alpha logn$ unless 
NP=P.
 \end{proof}
 \vspace{-0.1in}

The above theorem holds for general graphs. For special graphs, we have the following corollaries.

 \vspace{-0.05in}
\begin{corollary}
\label{cor-unitweight}
The CPMNC problem is NP-complete 
and cannot be approximated within $\alpha logn$
for some constant $\alpha < 1$ unless
NP=P even if  the graph is unit-weighted.
\end{corollary}

\begin{proof}
Note that in the above reduction, only those set nodes have
weight more than 1. We can change every such node $v$ with weight $c>1$ to a clique $Q_v$ of $c$ nodes, and 
connect each node $u$, which is originally connected to $v$ in the old graph, to each node in $Q_v$ in the new graph. 
Then the resulting graph is unit-weighted. Note that if one want to cut a set node in the
original graph, one must cut all the 
corresponding clique nodes in the new graph
to make the cut minimum.  So all the arguments in the proof of Theorem \ref{the-nc1}
still hold.   
\end{proof}

 \vspace{-0.15in}
 \begin{corollary}
 \label{cor-bipartite}
The CPMNC problem is NP-complete 
and cannot be approximated within $\alpha logn$
for some constant $\alpha < 1$ unless
P=NP  even if the graph is bipartite.
\end{corollary}

\begin{proof}
From the construction in the proof of Theorem \ref{the-nc1}, it is easy to see that 
if we shrink the two nodes connecting the neighboring
gadgets, the nodes in the graph can be partitioned
into two sets such that there is no edge in each set (i.e., the resulting graph is bipartite).
Thus the corollary is true for  bipartite  graphs. 
\end{proof}

 
\vspace{-5pt}

Next we show that the problem cannot
be approximated within any $poly(logn)$
ratio unless $NP \subseteq  DTIME(n^{poly(logn)})$.

\begin{theorem}
\label{the-nc2}
The CPMNC problem cannot be approximated within
a ratio of $log^kn$, for any positive $k$, unless $NP\subset DTIME(n^{poly(logn)})$.
\end{theorem}

The proof is based on Feige and Lovasz's 
two-prover one round interactive 
proof protocol~\cite{FL92} 
(abbreviated as $MIP(2,1)$) and
 Lund and Yannakakis's result on
the hardness of set cover~\cite{LY94}.
$MIP(2,1)$ consists of two provers $P_1$, $P_2$
and one verifier $V$. $Q_1$ and $Q_2$ are
sets of  possible
questions for $P_1$ and $P_2$, $A_1$ and $A_2$ are
 sets  of possible answers from $P_1$ and
$P_2$,  $\Sigma $ is the set of input
alphabet, and $R$ is a set of random seeds. The verifier
first computes a (polynomial time computable) function $f:\, \Sigma^n\times R \to Q_1\times Q_2$, to generate two
questions for $P_1$ and $P_2$. After receiving
the answers, $V$ computes a boolean predicate (also polynomial time computable)
on $ \Sigma^n\times R \times A_1\times A_2$ to decide acceptance or rejection. Notice that in the
protocol, the two provers can only agree with each other on some strategy pre-hand, and once
the execution of the protocol begins,
the two protocols can no longer communicate. This means that $P_1$ (or $P_2$) cannot see
the question for $P_2$ (or $P_{1}$), and the answer
from $P_2$ (or $P_{1}$). To achieve this securely, an oblivious protocol~\cite{Killian88} can
be used.

 The $MIP(2,1)$ protocol for NP has the following properties. 
 \begin{itemize}
\item  If the input $SAT$ instance $\phi$ is satisfiable,
then the provers always have a strategy to make the
verifier to accept.
\item If the input $SAT$ instance $\phi$ is not satisfiable,
then no matter what strategy the  provers use, the
verifier will accept with probability at most
 $1/n$, where $n$ is the input size.

\item All messages transferred in the 
protocol have length bounded by a 
polylog function. Also in the protocol, given an input 
instance, a random seed $r \in R$, and answer $a_1$ from $P_1$, 
there is a unique valid answer $a_2$ that the verifier will accept.
Additionally, in the construction, $ \mid Q_1 \mid = \mid Q_2 \mid $,
and for every $q_1 \in  Q_1$, there is
an equal number $ \mid R\mid /\mid Q_1\mid $ of $r$ that will generates
$q_1$. This is also the case for 
every $q_2 \in  Q_2$.
\end{itemize}

\vspace{-0.1in}
We now prove Theorem \ref{the-nc2}. 
\vspace{-0.1in}
\begin{proof} 
To prove the theorem, we first construct a graph.
Given a $SAT$ instance $\phi $, the graph will have a valid cut with weight
 at most $N  (\mid Q_1 \mid+ \mid Q_2\mid +1)$ if $\phi$ can be 
satisfied, where $N=2^{poly(logn)}$ is the total number of nodes
in the constructed graph.
If $\phi$ cannot be
satisfied, a valid cut will have weight at least
$n^\epsilon N(\mid Q_1 \mid + \mid Q_2 \mid +1) $, for some constant
$\epsilon$ and $0<\epsilon < 1$.

The graph has $\mid R\mid $ question gadgets ( $\mid R\mid $  also
has size $poly(logn)$~\cite{LY94} )
) with each denoted as 
$(r,q_1,q_2)$. Every $r\in R$ has a question gadget and a corresponding
question pair $(q_1,q_2)$.
We put every possible valid answer pair $(a_1,a_2)$
as two nodes in the gadget. All such 
answer pairs are put in the gadget
in parallel (see Figure \ref{fig:mip}). The gadget also
has a backdoor node $g_1$.
This backdoor node is  connected
to  $t$ through an intermediate node 
$g_2$,  which has a very large weight, say,
$n N(\mid Q_1 \mid + \mid Q_2 \mid +1)$.
 All nodes in a gadget have weight $1$. We also have $ \mid Q_1\times   A_1 \mid  + \mid Q_2\times  A_2 \mid $ answer
nodes with each of them associated with weight $N$.  Every answer
node $(q_1,a_1)$ (or $(q_2,a_2)$ ) is connected
to the gadget node $a_1$ (or $a_2$) if
the gadget is $(r,q_1,q_2)$.
Note that an answer node may be connected
to multiple gadget nodes. Finally, every
answer node is connected to $t$, and the
gadgets are connected sequentially with
$s_1$ and $s_2$ at two ends (see Figure
~\ref{fig:mip}).
 
%

Let  $c(q_1)$ (or $c(q_2)$) be the number of nodes selected in a cut from all those answer nodes corresponding to the same question 
$q_1$ (or $q_2$). If $\phi$ can be satisfied, then we can find
a cut with weight at most $N(\mid Q_1\mid  + \mid Q_2 \mid + 1)$ .
This is because the prover $P_1$ (or $P_2$) can have
a valid answer $a_1$ ( or $a_2$) for any question $q_1$ (or $q_2$), and
if we choose these $(q_1,a_1)$ and $(q_2,a_2)$ ($ \mid Q_1 \mid + \mid Q_2 \mid $ nodes in
total) answer nodes in the cut, we can have
a valid cut with weight at most  $N( \mid Q_1\mid +\mid Q_2 \mid +1 )$. 

If $\phi$ cannot be satisfied, we then have two cases to consider.

{\bf Case 1:} For some valid $(r,q_1,q_2)$,
there is no valid answer pair $(a_1,a_2)$. In this
case, to find a valid cut, one must choose 
node $g_2$ in the $(r,q_1,q_2)$ gadget which has weight  $nN( \mid Q_1\mid +\mid Q_2 \mid +1 )$. 

{\bf Case 2:} For every valid  $(r,q_1,q_2)$, 
there always exists at least one valid answer pair $(a_1, a_2)$.
In this case, 
let $p$ be percentage of those $r$ whose corresponding $q_1$ and $q_2$ have $c(q_1)+c(q_2) \le n^{\epsilon_1}$, where $0<\epsilon_1<1/2$ is a small positive real number.  Then we have $p< (n^{2\epsilon_1-1}) $.
 To see this, suppose $p> (n^{2\epsilon_1-1}) $. Then for a question node $q_1$, prover $P_1$ can randomly 
select one of the answers $a_1$ such that $(q_1,a_1)$
is in the cut, prover $P_2$ can randomly  select
one of answers $a_2$ such that $(q_2,a_2)$
is in the cut, and  $(a_1,a_2)$ is a valid answer
with probability at least $n^{-2\epsilon_1}$. Thus the total
probability that the provers  will
get a valid answer is 
at least $ \frac{p}{n^{2\epsilon_1}} $,
which is greater than $n^{-1}$, a
contradiction. 
Hence we have
 $ \displaystyle\Sigma_rc(r) \geq (1-n^{2\epsilon_1-1} )\mid R \mid  n^{\epsilon_1},$
where $c(r)= c(q_1)+c(q_2)$  and $q_1$ and $q_2$ are the
 queries corresponding to seed $r$.  From this, we immediately have
 $ \displaystyle\Sigma_rc(r) =  \displaystyle\Sigma_r(c(q_1) + c(q_2)) 
   = \frac{\mid R \mid }{\mid Q_1\mid} 
         (\displaystyle(\Sigma_{q_1\in Q_1}c(q_1) +(\Sigma_{q_2\in Q_2}c(q_2) ) 
  =   \frac{\mid R \mid }{\mid Q_1\mid} C,$
  where $C$ is the total number of answer nodes in the cut. Combining the above two, we have 
   $\frac{\mid R \mid }{\mid Q_1\mid} C  \geq (1-n^{2\epsilon_1-1} )\mid R \mid  n^{\epsilon_1}. $ 
   This implies that 
   $C \geq\mid Q_1\mid (1-n^{2\epsilon_1-1} ) n^{\epsilon_1} > n^{\epsilon }(\mid Q_1 \mid +\mid Q_2 \mid +1 ),$
 where $0<\epsilon<\epsilon_1$ (Note that here we use the fact that $\mid Q_1 \mid = \mid Q_2 \mid$) .

Thus the total weight of the cut will be  larger than $n^\epsilon N(\mid Q_1 \mid +\mid Q_2 \mid +1 )$.
Note that $N=2^{poly(logn)}$, and for any positive number $k$,
 $n^\epsilon>log^kN$ for all sufficiently
large $n$. This proves our assertion.
\end{proof}
\vspace{-0.05in}

%


\section{ CPMC in Planar Graphs}
\label{sec-planar}
\vspace{-0.1in}

In this section we present polynomial time solutions to CPMEC in planar graphs and  CPMNC in some special planar graphs.

\vspace{-0.05in}
\begin{theorem}
\label{the-planar2}
If the graph $G$ is planar and the source node
$s_1$ and the partner node $s_2$ are  in the same face,
then the CPMNC problem can be solved in polynomial
time.
\end{theorem}
\vspace{-0.1in}

\begin{proof}
If $s_1$ and $s_2$ are  in the same face,
we can find a planar embedding of $G$ such that 
$s_1$ and $s_2$ are on the boundary of the embedding (i.e, on the outer face). 
It is easy to see that after removing the connectivity
preserving minimum cut separating $s_1$ and $s_2$ from $t$,
$s_1$ and $s_2$ are still connected by one
of the two boundary paths between 
$s_1$ and $s_2$ (see Figure ~\ref{fig:pla1}, where
$s_1$ and $s_2$ are connected by either the 
path $s_1,A_1,\ldots,A_k,s_2$ or
the path $s_1,B_1,\ldots ,B_m,s_2$). Thus we can use the following
algorithm to solve the problem.

\begin{enumerate}
\item Add a dummy node $D$, and connect $D$ to nodes $s_1,A_1,\ldots,A_k,s_2$.  Set the
weight  of nodes $s_1,A_1, \ldots , A_k,s_2$ to infinity.
\item Compute the minimum cut between $D$ and $t$. Let $x_{1}$ be the
weight of this cut.
\item  Remove all the previously added edges, connect $D$ to nodes
$s_1,B_1,\ldots,B_m,s_2$,  and set the weight  of nodes $s_1,B_1,\ldots,B_m,s_2$
to infinity.  
\item Compute the minimum cut between $D$ and $t$.  Let $x_{2}$ be the
weight of this cut.
\item Choose the smaller one between $x_1$ and $x_2$ as solution. 
\end{enumerate}
Obviously, the above algorithm runs in polynomial time and generates the optimal solution. Thus the theorem follows. 
\end{proof}

 \begin{figure}
\centering
\includegraphics[width=2.5in]{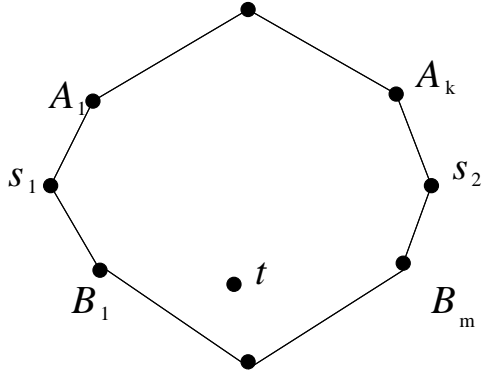}
\caption{An example illustrating Theorem \ref{the-planar2}.}
\label{fig:pla1}
\end{figure} 
 

 Next we show that CPMEC in planar graphs has polynomial 
 time solutions.
 
 First we introduce a perturbation technique that is crucial for the algorithm  of the CPMEC problem.
 Let $G=\{V,E\}$, $V=\{v_1,\ldots,v_n\}$, $E=\{e_1,\ldots,e_m\}$ be an undirected graph with each edge $e_i$
 associated with a non-negative weight (or cost) $c_i$. We use a new weighting function
 $c_i'=c_i+\epsilon_i$ for each edge $e_{i}$, where $\epsilon_i$
 is a small positive  perturbation number. The perturbation numbers
 are assigned in a way that no two set of edges have the same total
 weights. Note that such a perturbation always exists. For example, we can first arbitrarily order all edges and add a small value to the weight (assuming to be an integer) of each edge which is $10^{-r}$, where $r$ is the rank of the edge in the order. In this way, any two cuts (or more generally, two subsets of edges) in $G$ will 
 have different weights unless they are completely identical. Based on this property, we
 have the following observations. (1) Any cut  is unique. (2) Given
 any node $v\in V$, let $C_{v}$ be the connected component containing $v$ and resulting from the minimum edge cut between $v$ and $t$. Then all 
 nodes in $C_v$ can be uniquely determined due to the  perturbation technique. 
 
  Now consider the CPMEC between nodes $s_1$, $v$, and $t$, which is also unique. This CPMEC cuts $G$
 into two connected components, and let $C_{s_1,v}$ be 
 the one containing $s_1$ and $v$.
 If $C_{s_1,v_1} = C_{s_1,v_2}$ for two different nodes $v_{1}$ and $v_{2}$, we say that $v_{1}$ and $v_{2}$ are {\em connectivity preserving equivalent (CPE)}. 
 We can  classify all nodes in $G$ into multiple 
 CPE classes.
 
 Starting from node $s_1$ in the new graph, we can compute
 $C_{s_1,v}$ for every node $v\neq t$ in the graph
 using Algorithm~\ref{alg:CPMEC}.
 
   \begin{algorithm}
   
       Fix a planar embedding of $G$ with $t$ being a node in the outer surface.
 
       Let $S=\{s_1\}$, $C_{s_1,s_1}= \{s_1\} $, and $C_{ep}(s_1,s_1,t)=C_e(s_1,t)$;
 
       \While{$s_2 \notin S$} {
       \For{every neighbor  $v$ ($v\notin S$)
        of a node $ s \in S$}  {
        Compute the minimum cut that separates
             $v$ and all nodes in $C_{s_1,s}$ from $t$; let
 the connected component containing $v$ and $C_{s_1,s}$ be
  $C_{s_1,v}$ ; 
  
       let the weight of  the cut be $u(v,s)$; }
 
       Find the pair of $v$ and $s$ with the minimum $u(v,s)$; denote them as
       as $v*$ and $s*$;
 
       $S= S\bigcup C_{s_1,v*} $; 
  
       \For{every node $v'$ in  $C_{s_1,v*}$ } {
              
            \If{$v'\notin S$ } {
 
               $C_{s_1,v'} = C_{s_1,v*}$;
 
               $C_{ep}(s_1,v',t)=u(v*,s*)$;
 
            }
          
        }
     } 
      Output $C_{ep}(s_1,s_2,t)$.
    \caption{CPMEC Algorithm for Planar Graphs}
 \label{alg:CPMEC}
 \end{algorithm}     
    
  In Algorithm~\ref{alg:CPMEC},  $C_{ep}(s_1,v,t)$ is the value of  CPMEC between
 $s_1$, $v$, and $t$. $C_e(s_1,t)$ is the value  of the minimum
 cut between $s_1$ ant $t$. The minimum cut separating $v$ and $C_{s_{1},s}$ from $t$ can be reduced to computing a minimum cut of two nodes $D$ and $t$, where $D$ is a dummy node connecting to $v$ and every 
 node in $C_{s_{1},s}$ with an edge of infinity weight. The 
idea of the algorithm is similar to the idea of Dijkstra algorithm~\cite{Dijkstra59} for
shortest path. Though the idea is straightforward, the proof is highly non-trivial. It is
also intriguing that
the proof does not work for general graphs. It would be interesting to 
classify the types of graphs that the  algorithm can find the optimal cut.

  It is easy to see that this is a polynomial time algorithm for finding the
 CPMEC between $s_1$, $s_2$ and $t$.
 We have the following observation. 
 In each iteration (the While loop),  the algorithm finds 
 a node $v$ with the minimum $C_{ep}(s_1,v,t)$ among
 all nodes not in $S$. This implies that the  CPMEC found in each
  iteration is non-decreasing, and for any node $v* \notin S$, $C_{ep}(s_1,v*,t) > C_{ep}(s_1,v,t) $
  for any node $v \in S$. 
 
The basic idea for proving the correctness of the algorithm is to show that
if the cut obtained at any step of the algorithm
is not the actual CPMEC, then we will have a contradiction. 
 The contradiction can be obtained
from the fact that for any two ``neighboring'' CPMECs, there is no hole
completely surrounded by the two CPMECs. 

We need the  
 following lemmas for the proof.
 
 \vspace{-0.1in}
 \begin{lemma}
 \label{lemma:containment}
Let $\alpha$ be a node in $C_{s_1,v*}$ and $ Y= C_{s_1,\alpha} \cap C_{s_1,v*} $. If 
the connected component of $Y$ containing $s_1$  also contains $\alpha$, then
$C_{s_1,\alpha}$ is completely contained in $C_{s_1,v*}$.
 \end{lemma}

\begin{proof}
 Suppose that $C_{s_1,\alpha}$ is not completely contained 
in $C_{s_1, v*}$. Then consider the region $Z=C_{s_1,\alpha}\backslash C_{s_1, v*}$
(see the shadowed area in Fig.~\ref{fig:contain}). The existence of 
region $Z$ violates the uniqueness of a cut value. This means that $C_{s_1,\alpha}$ must
be completely contained in  $C_{s_1, v*}$.
\end{proof}

 \vspace{-0.15in}
 \begin{lemma}
 \label{lemma:share-cut1}
 Let $C_{1}$ be the CPMEC separating $s_{1}$ and $A$ from $t$, and $C_{2}$ be a different CPMEC separating $s_{1}$, $A_{1}$ from $t$
 such that $A \in C_{s_1,A_1}$ and  none of the two sets $C_{s_1,A}$
 and $C_{s_1,A_1}$ completely contains the other.  Then there exists no hole that is completely
surrounded by $C_{s_1,A_1}$ and  $C_{s_1,A}$.
 \end{lemma}
 
\begin{proof}
 First, from the definition, we know that both $C_{s_1,A}$ and $C_{s_1,A_1}$ 
 contain $s_{1}$ and $A$. The intersection of $C_{s_1,A}$ and $C_{s_1,A_1}$  
 must have one connected component containing $s_{1}$ and another connected component containing $A$. 
 Otherwise, if there is only one connected component, then we can 
 either expand or shrink the boundary of one of $C_{s_{1},A}$ and $C_{s_{1},A_{1}}$ and still form a CPMEC. This violates the fact that each CPMEC cut is unique. 
 
If there exists a hole that is completely surrounded by the two CPMECs,
then for the boundary of the hole, we can divide the boundary into
 3 types of segments: The first type of segments is in the boundary
of $C_{s_{1},A}$ but not boundary of  $C_{s_{1},A_{1}}$. We denote the total length
of this type of segments as $L$. The second type of segments is in the boundary
of $C_{s_{1},A_1}$ but not boundary of  $C_{s_{1},A}$. We denote the total length
of this type of segments as $L_1$. The third type of segments is in the boundary
of $C_{s_{1},A_1}$  and  $C_{s_{1},A}$. By the uniqueness of the weight value,
we have either $L>L_1$ or  $L_1>L$. If $L>L_1$ then we combine  CPMEC $C_{s_{1},A}$
with the region inside the hole. Now we can get a smaller CPMEC, because 
the new cut will decrease by a value of $L$ and increase  by a value of $L_1$,
and the overall effect is that the value of the cut will decrease by at least
$L-L_1$. This is a contradiction.
If $L<L_1$ then we combine CPMEC $C_{s_{1},A_1}$
with the region inside the hole. Now we can get a smaller CPMEC, because 
the new cut will decrease by a value of $L_1$ and increase  by a value of $L$,
and the overall effect is that the value of the cut will decrease by at least
$L_1-L$. This is also a contradiction. So there is no hole that is completely
surrounded by $C_{s_{1},A}$ and $C_{s_{1},A_1}$.
\end{proof}

\vspace{-0.19in}
 \begin{lemma}
 \label{lemma:share-cut2}
 Let $C_{1}$ be the CPMEC separating $s_{1}$ and $A$ from $t$, and $C_{2}$ be a different CPMEC separating $s_{1}$, $A_{1}$ from $t$ such that
 $A$ and $A_1$ are in the same face of the graph with $ A_1 \notin C_{s_1,A}$
 and $A \notin C_{s_1,A_1}$.
 Let  $Y_1$ be the graph induced by the set of nodes
 in $C_{s_1,A}$ but not in $C_{s_1,A_1}$, and  $Y_2$ 
 be the graph induced by the set of nodes in $C_{s_1,A_1}$ but not in $C_{s_1,A}$. 
  Then there exists no hole that is completely
surrounded by $C_{s_1,A_1}$ and  $C_{s_1,A}$.
 \end{lemma}

\begin{proof}
  The argument is similar as that in Lemma~\ref{lemma:share-cut1}. The only difference
is that the intersection of $C_{s_1,A}$ and $C_{s_1,A_1}$  
 must have one connected component containing $s_{1}$, but $A \notin C_{s_1,A_1}$ and $A_1 \notin C_{s_1,A}$.
If there exists a hole that is completely surrounded by the two CPMECs,
then we can also 
  expand  the boundary of one of $C_{s_{1},A}$ and $C_{s_{1},A_{1}}$ 
to incorporate the region in the hole and  form a smaller CPMEC,
which is a contradiction.  
 \end{proof}

\vspace{-0.1in}
 \begin{lemma}
 \label{lemma:hole}
 At any time point during the execution of the algorithm, there is no hole (i.e., a missing subgraph in the embedding of $G$) in $S$.
 \end{lemma}

\begin{proof}
We  use mathematical induction to prove this. Assume that
before node $v$ is added into $S$, there is no hole in $S$.
Suppose after $v$ is added, a hole forms in $S$, as shown in Fig.~\ref{fig:hole}.
Consider the node $s$ that is chosen in the algorithm during the iteration that $v$ 
is added (in the line "find the pair of $v$ and $s$ with the minimum $u(v,s)$").
Now we can see that there must exist a hole between  $C_{s_1,s}$  and  $C_{s_1,v}$.
Since  $C_{s_1,s}$ is contained
in  $C_{s_1,v}$, we can see that it is impossible that any hole can exist 
between $C_{s_1,s}$ and $C_{s_1,v}$, using a similar argument 
 in the proof of Lemma~\ref{lemma:share-cut1}.
Thus, there is no hole in $S$ at any step of the algorithm.
\end{proof}  


 We use mathematical induction to show that the algorithm finds the 
 CPMEC between $s_1$, $s_2$ and $t$. The base case is the initial step which is obviously true as $C_{ep}(s_{1},s_{1},t)=C_{e}(s_{1},t)$. For the induction hypothesis, we assume that
 all nodes added to $S$ in previous iterations have their $C_{ep}$values correctly computed (i.e., equal to their true CPMEC value).  Also, for any node $i$ added into $S$ in some iteration after node $j$, the  $C_{ep}$ value 
 for node $j$ is less than the  $C_{ep}$
 value for node $i$.  
 Now consider the iteration when node $v*$ is added to $S$. 
 
   Suppose that $C_{ep}(s_1,v*,t)$ computed in this iteration
 is not the true CPMEC. 
 Let $C'_{ep}(s_1,v*,t)$ be the
 true CPMEC, 
 and $C'_{s_1,v*}$
 be its connected component containing $s_1$ and $v*$.
  Consider the intersection $U_1=C'_{s_1,v*} \cap S$, where $S$ is the 
 set of added nodes before this iteration.
 Let $U_{2}$ be the connected component of $U_1$ containing $s_1$ and $U_3 =  C'_{s_1,v*} \setminus U_2$. 
By Lemma~\ref{lemma:hole}, we know that there should be no hole in $U_2$.
 Let $U_4$ be the set of nodes in $U_2$ connected to $U_3$. Let $\alpha$ be any node in  $U_4$
 connected to  a node $\beta$ in $U_3$ and $\Gamma=U_2\cap C_{s_1,\alpha}$.  Denote the
 connected component of $\Gamma$ that contains $s_1$ as $\Gamma_1$. If
 $\alpha \in \Gamma_1$, then we must have $ C_{ep}(s_1,\alpha,t) \leq C(U_2)$, 
 where $C(U_2)$ is the cut value of  $U_2$ (i.e., the cut between $U_{2}$ and $G\setminus U_{2}$). 
 
 Now we have $C_{ep}(s_1,\beta,t) \leq C(U_3) + C_{ep}(s_1,\alpha,t) - C(U_2,U_3)$, where $C(U_3)$ is the cut value of region $U_3$, and $C(U_2,U_3)$
 is the cut value between $U_2$ and $U_3$.    But,
    $ C(U_3) + C_{ep}(s_1,\alpha,t) - C(U_2,U_3)
         \leq C(U_3) + C(U_2) - C(U_2,U_3) $
 and 
        $ C(U_3) + C(U_2) - C(U_2,U_3) = C'_{ep}(s_1,v*,t)$. 
 Thus we have
   $ C_{ep}(s_1,\beta,t) \leq C'_{ep}(s_1,v*,t)$.  
 This means that $\beta$ should be added into $S$ before $v*$,
  according to
 the induction hypothesis. This is a contradiction. Hence we know that
 for any  $\alpha$ in  $U_4$, $\alpha \notin \Gamma_1$.
 
  If $U_4$ has only one node
 $\alpha$ which connects to $\beta$ in $U_3$ (see Fig.~\ref{fig:case1}),
  then  
 $ C_{ep}(s_1,\beta,t) \leq C(U_3) + C_{ep}(s_1,\alpha,t) - C(U_2,U_3) $.
 Thus we also have
   $C_{ep}(s_1,\beta,t) \leq C'_{ep}(s_1,v*,t) $, 
         which is a contradiction. 
 
  \begin{figure}
 \centering
 \includegraphics[height=2.5in]{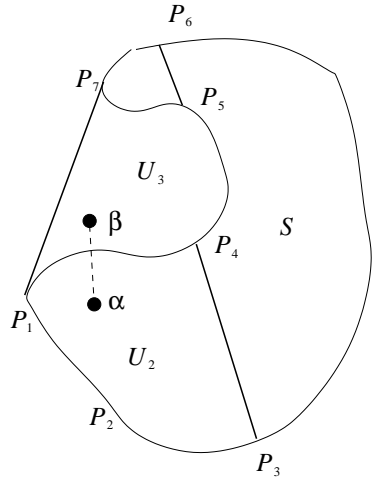}
 \caption{The first case of $U_4$.}
 \label{fig:case1}
\end{figure}

\begin{figure}
 \centering
 \includegraphics[height=2.5in]{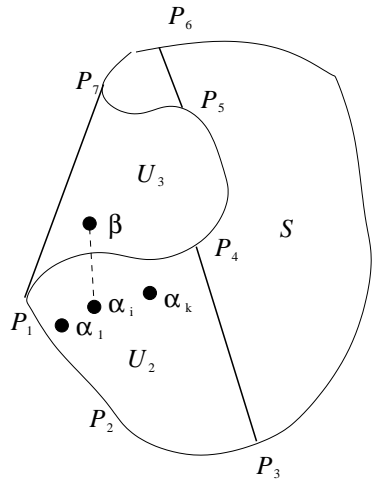}
 \caption{ The second case of $U_4$.}
 \label{fig:case2}
 \end{figure}

 If $U_{4}$ has more than one node, denote them as $\alpha_1,\ldots,\alpha_k$ (see Fig.~\ref{fig:case2}), then
 we have the following Lemma. 
 
 \begin{lemma}
 \label{lemma:orientation}
  There exists at least one 
  $\alpha \in \{\alpha_1,\ldots,\alpha_k\}$ such that  $ C_{s_1,\alpha} \subseteq U_2$.  
 \end{lemma}
  
\begin{proof}
First we show that there exists no node $\alpha$ in $S$  such that  
 the boundary of $C_{s_1,\alpha}$ completely crosses $C'_{s_1,v^{*} }$ and 
 partitions the region occupied by $C'_{s_1,v^{*} }$ 
 into two (or more) connected regions $K_1$ and $K_2$ outside of  $C_{s_{1},\alpha}$ 
and another connected region inside $C_{s_{1},\alpha}$ containing both  
$\alpha$ and $s_1$ (see Fig.~\ref{fig:below}),
where segment $PQ$ is the crossing, and the thick curve denotes the
boundary of $C_{s_1,\alpha}$ and the thin curve denotes the boundary
of $C'_{s_1,v^{*} }$. 
If such $\alpha$ exists, then we can replace $K_2$ with $C_{s_1,\alpha}$
in $C'_{s_1,v^{*} }$ to obtain a smaller CPEMC for $v^{*}$.  This is because the $C_{ep}(s_1,\alpha,t)$ is  less than
the value of the cut surrounding  $K_2$ since
$C_{s_1,\alpha}$ is a CPEMC and $K_2$ contains
both $s_1$ and $\alpha$. 
This contradicts the fact that
  $C'_{s_1,v^{*} }$ is the minimum cut.
We call this property as the non-crossing property. 

 Now we prove the lemma by contradiction. Suppose that any node $\alpha$ in  $U_4$
 has the property that $C_{s_1,\alpha} \not\subseteq U_2$.
 First we classify the nodes in $U_4$ into two classes, left nodes and right nodes.
 For any node $\alpha \in U_4$, consider the components that is in $C_{s_1,\alpha}$
 but not in $C'_{s_1,v^{*} }$, and  denote it as $Y_{\alpha}$. 
 Denote the connected component in $C_{s_1,\alpha} \cap C'_{s_1,v^{*} }$ that 
 contains $s_1$ as $X_{\alpha}$, the connected component in $C_{s_1,\alpha} \cap C'_{s_1,v^{*} }$ that 
 contains $\alpha$ as $H_{\alpha}$, and
   the connected
 component in  $Y_{\alpha}$ that is connected to $H_{\alpha}$ as $Z_{\alpha}$. 
 If $Z_{\alpha}$ is on the left hand side
 of $X_{\alpha}$, $\alpha$ is a left node, and otherwise
 a right node (assuming that we are standing at $s_{1}$ and facing against the portion of  $C_{s_{1}, \alpha}$ outside of $C'_{s_{1},v^{*}}$). Fig.~\ref{fig:left} and Fig.~\ref{fig:right} show these two cases.
 If all the nodes in $U_4$ are left nodes, consider the rightmost node $\alpha$ in 
 $U_4$ and the region  $C_{s_1,\alpha}$, as shown in Fig.~\ref{fig:allleft}. 
In this case $C_{s_1,\alpha}$ will completely cross
$C'_{s_{1},v^{*} }$ and both nodes $\alpha$ and $s_1$ are in the same side
of the crossing, this violates the non-crossing property. Similarly if all nodes are right nodes, we will
 also get a contradiction. So we must have two nodes $\alpha_i$ and   $\alpha_j$ that
 are in the same face with $\alpha_i$ being a left node, and $\alpha_j$
 being a right node. Now we have three cases to consider,
(1) $\alpha_i \in C_{s_1,\alpha_j}$,
(2)  $\alpha_j \in C_{s_1,\alpha_i}$,
and (3)   $\alpha_i \notin C_{s_1,\alpha_j}$ and 
$\alpha_j \notin C_{s_1,\alpha_i}$. In the last
case, suppose that $C_{s_1,\alpha_i}$ does not 
completely cross $C'_{s_{1},v^{*}}$
with $\alpha_i$ and $s_1$ in the same side of the
crossing. Now let $P_{0}$ and $P_{1}$ be  the first two 
intersection points
between $C_{s_1,\alpha_i}$ and $C'_{s_1,v^{*} }$
while $C_{s_1,\alpha_i}$ goes out of $C'_{s_1,v^{*} }$,  and $P_{3}$ and $P_{4}$ be the first two 
intersection points
between $C_{s_1,\alpha_j}$ and $C'_{s_1,v^{*} }$
while $C_{s_1,\alpha_j}$ goes out of $C'_{s_1,v^{*} }$. 
   By Lemma~\ref{lemma:share-cut2}, we know that there is no hole that 
is completely  surrounded by $C_{s_1,\alpha_j}$ and $C_{s_1,\alpha_i}$.
We can see that  $C_{s_1,\alpha_i}$ must have the boundary
segment $P_1P_2$  as shown in  Fig.~\ref{fig:twist} (Here $P_2$ is one of the 
neighboring face of $\alpha_i$ in $G$) such that no point in segment $P_1P_2$
is outside the boundary $C_{s_1,\alpha_j}$ (but can be in the boundary), 
otherwise there exists a hole that is
completely  surrounded by $C_{s_1,\alpha_j}$ and $C_{s_1,\alpha_i}$.
This means that $C_{s_1,\alpha_j}$  completely crosses $C'_{s_1,v^{*}}$, which is a violation of the 
non-crossing property.

For the first two cases,
 we can easily see that by Lemma~\ref{lemma:share-cut1},
either $C_{s_1,\alpha_i}$ or $C_{s_1,\alpha_j}$ 
will completely cross  $C'_{s_1,v^{*} }$ following a similar argument. Thus we have  a
 contradiction for both cases. This implies that the lemma holds.
  \end{proof}

\begin{figure}
\centering
  \includegraphics[width=2.5in]{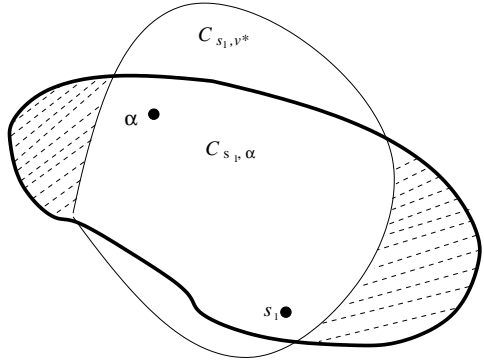}
  \caption{Illustration for containing relationship.}
  \label{fig:contain}
\end{figure}

\begin{figure}
  \centering
  \includegraphics[width=2.5in]{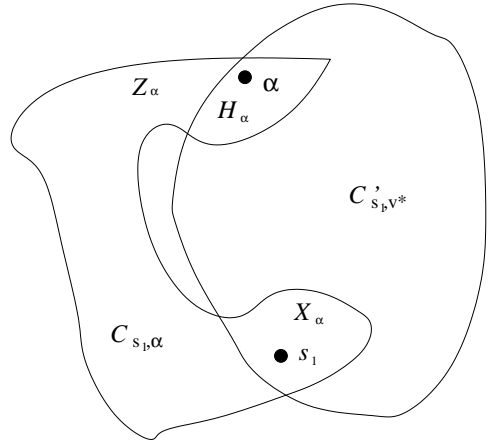}
  \caption{ Illustration for left node.}
  \label{fig:left}
\end{figure}

\begin{figure}
\centering
  \includegraphics[width=2.5in]{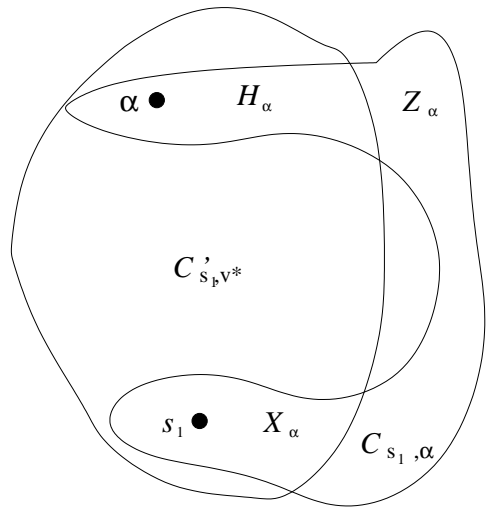}
  \caption{Illustration for right node.}
  \label{fig:right}
\end{figure}

\begin{figure}
  \centering
  \includegraphics[width=2.5in]{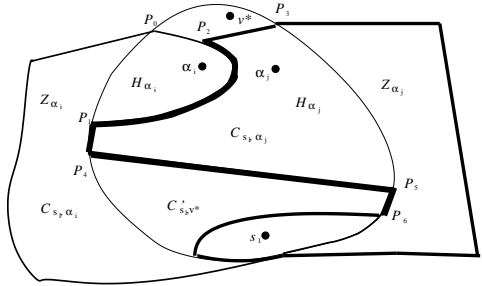}
  \caption{ Illustration for a twisted cut.}
  \label{fig:twist}
\end{figure}

\begin{figure}
\centering
  \includegraphics[width=2.5in]{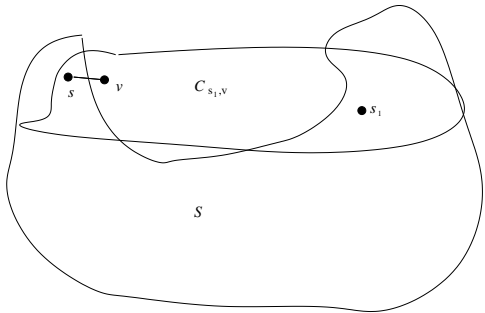}
  \caption{Illustration for holes in $S$.}
  \label{fig:hole}
\end{figure}

\begin{figure}
  \centering
  \includegraphics[width=2.5in]{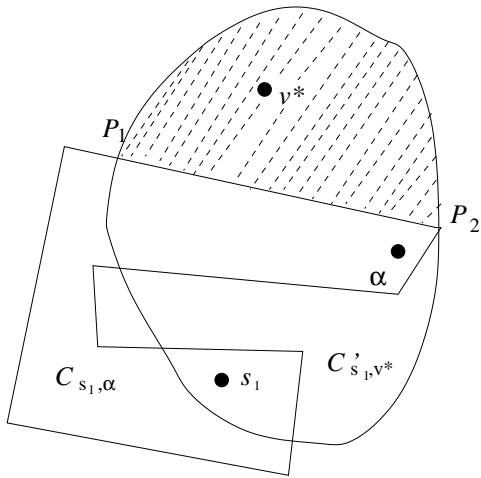}
  \caption{ Illustration for all left nodes.}
  \label{fig:allleft}
\end{figure}

  \begin{figure}
 \centering
 \includegraphics[height=2.5in]{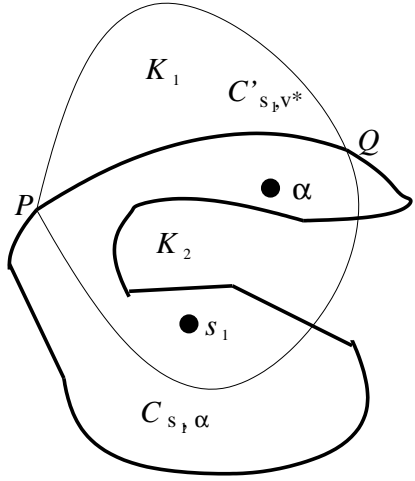}
 \caption{ Illustration for the non-crossing property. }
 \label{fig:below}
 
 \end{figure} 

 From this lemma, we know that there exists such an $\alpha$
satisfying the condition in the Lemma. Let $\beta$ be the node in $U_{3}$ connected to $\alpha$.
  Suppose that the 
  node in $U_3$ that is connected to $\alpha$ is $\beta$. 
 Then $\beta$ must be added into $S$ before $v*$, which is a contradiction.
 Thus we have the correctness of the algorithm.
 
 For the running time of the algorithm, it is easy to see that there are at most $O(n)$ iterations in the while loop and 
 each iteration takes $O(n^2T_{mc})$ time, where $T_{mc}$ is the
 time for computing the minimum cut for two nodes. Thus, the total time complexity is polynomial.

\noindent{\bf Remarks:} The above time bound is mainly for showing CPMEC is in $P$ for planar graphs. We leave it as future work for designing faster algorithms.

 \vspace{-0.05in}
 \begin{theorem}
 \label{the-eplanar}
 CPMEC  in planar graphs can be solved in polynomial time.
 \end{theorem}

Next we present another algorithm for a special case of the CPMEC problem
in which $s_1$ and $t$ are in the same face of a planar graph. The idea of  the algorithm is completely different
 from the above algorithm for general planar graphs and it has an interesting relationship with another
 problem, called {\em Location Constrained Shortest Path
(LCSP)}. LCSP  finds applications in VLSI design and robotics.  LCSP  corresponds to the 
CPMEC problem in its dual planar graph.


\vspace{-0.05in}
  \begin{definition}[LCSP]
Let $G=(V,E)$ be a planar graph with a fixed embedding and each edge $e_{i} \in E$ associated with a weight $w_i$. Let $A$, $B$ be two nodes on the boundary of the embedding (without loss of generality, assume that the segment connecting $A$ and $B$ is horizontal) and $C$ be an interior face.  Find a shortest path from $A$ to $B$ along the interior nodes of $G$ with $C$ staying above the path.
  \end{definition}
\vspace{-0.15in}

\begin{theorem}
\label{the-lcsp}
The  LCSP problem can be solved in polynomial time.     
\end{theorem}
\vspace{-0.1in}

\begin{proof}
We first show that given two nodes $A$ and $B$ in the outer face, and another interior edge $UV$ of the
graph, we can find a shortest interior path from
$A$ to $U$, then passing $V$ through edge $UV$, and reaching $B$ in polynomial time
(Fig.~\ref{fig:dummy1} ).
That is, we can find the shortest path that passes a 
specified edge along a specified direction.
 This can be done  in polynomial time based 
on the algorithm for two node-disjoint paths with minimum sum length 
when the end points of the two paths are in two faces, as shown in~\cite{KS10,VS08}.

Next, we show that we can solve the LCSP problem in polynomial time in an iterative fashion.

For this,  we first make an observation that for any path to keep the inner face $C$ above the path, a shortest path cannot pass 
any of its edges in the clockwise direction (with respect to face $C$).
This is shown in Fig.~\ref{fig:clockwise}. If we want to find a shortest path
from $A$ to $B$ and keep face $C$ above the path, then the path cannot pass through edge $DE$ along the direction from $D$ to 
$E$; otherwise the path will not keep $C$ above it. Hence a valid path will only pass through the edge of $C$ (if it does)
in a couterclockwise direction. On the contrary, if an interior path passes through an edge of $C$ in a couterclockwise direction,
then  $C$ must be above the path. 

Thus to find the LCSP, we conduct the following computation for each iterative step. For every edge of the face $C$, we first specify its 
counterclockwise direction, and compute the shortest path from $A$ to $B$ passing the edge in  the counterclockwise direction. We store  
these shortest paths in some data structure. Then, we remove all edges of face $C$, and all remaining degree-one
nodes. 

Now the face $C$ is enlarged.  We can repeat the above iterative step, and store the computed shortest paths in the same data structure. 

Finally, face $C$ will reach the boundary of the graph, in which case we can find the shortest path trivially. Now we can choose
 the shortest path among all stored paths which is the desired shortest path. 

\begin{figure}
\centering
  \includegraphics[width=2.5in]{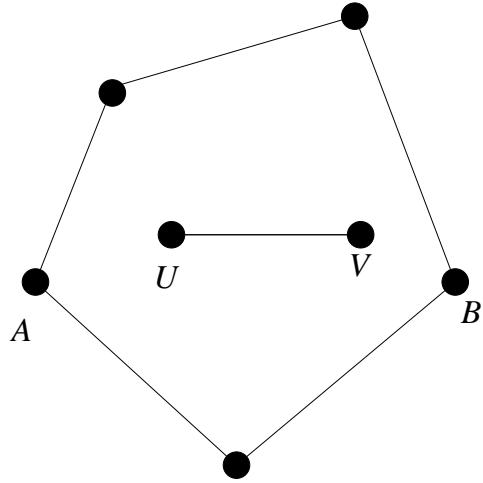}
  \caption{Shortest path with one specified intermediate edge.}
  \label{fig:dummy1}
\end{figure}
\begin{figure}
  \centering
  \includegraphics[width=2.5in]{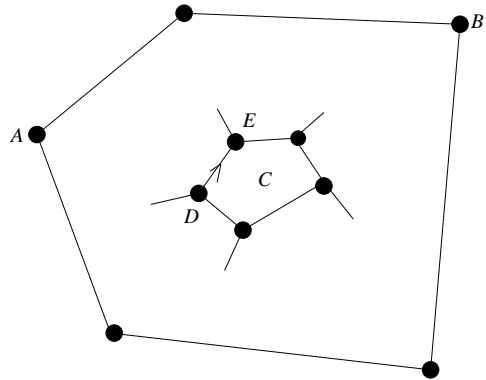}
  \caption{ Clockwise edge of the face.}
  \label{fig:clockwise}
\end{figure}
 \end{proof}

\begin{theorem}
\label{the-ecplanar}
There exists a polynomial time algorithm for CPMEC problem in planar graphs
when $s_1$ and $t$ are in the same face, based 
on the algorithm for LCSP.
\end{theorem}

\begin{proof}
We can see that any shortest path
between two nodes of the boundary in the dual graph will pass through
zero  or two edges in every face, and any shortest path  
between two boundary nodes separated by $s_1$ and $t$
and keeping $s_2$ above the path will keep $s_1$ and $s_2$ connected
in the original graph if we convert the shortest path into a cut.

Now we  can try all  pairs of nodes
that are separated by $s_1$ and $t$ on the boundary of the dual graph; 
then find the shortest that keeps $s_2$ above the path, using
the algorithm in Theorem \ref{the-lcsp}. We choose the
shortest path among all these trials, which is exactly
the minimum edge cut in the original graph (it is easy to see that 
this cut will keep $s_1$ and $s_2$ connected ). 
\end{proof}

\section{ Remarks and Future Work}
\label{sec-fw}
 
Several issues related to the CPMC problem remain open and will be future research directions.
 First of all, we conjecture that the connectivity preserving minimum
 node cut problem cannot be approximated
 within $n^{\epsilon}$ for some $\epsilon < 1$, or even for any   $\epsilon < 1$. Secondly, the
 hardness of CPMEC problem  (3-node case) for undirected graphs is still
 open. As we mentioned earlier, the 
 hardness proofs for the CPMNC problem cannot be directly extended to
 the CPMEC problem. Thus new proving techniques are needed.
 Thirdly,  we believe that more efficient
precise or approximation algorithms exist for some special graphs and it will  be another future research direction.





\bibliographystyle{elsarticle-num}
\bibliography{cpmc}







\end{document}